\documentclass[12pt, letterpaper]{article}

\usepackage[a4paper,margin=1in]{geometry}

\usepackage{booktabs} 
\usepackage[T1]{fontenc}    
\usepackage{authblk}
\usepackage{amssymb}
\usepackage{amsthm}
\usepackage{amsmath}
\usepackage{latexsym}
\usepackage{color}
\usepackage{verbatim}
\usepackage[english]{babel}
\usepackage[ruled,linesnumbered]{algorithm2e} 
\usepackage{subcaption}
\usepackage{graphicx}

\usepackage{algorithmic}
\usepackage[square, numbers]{natbib}
\usepackage[hidelinks]{hyperref}
\DeclareMathOperator*{\argmax}{arg\,max}

\newtheorem{theorem}{Theorem}
\newtheorem{lemma}{Lemma}
\newtheorem{assumption}{Assumption}
\newtheorem{definition}{Definition}
\newtheorem{example}{Example}
\newtheorem{remark}{Remark}
\newtheorem{corollary}{Corollary}

\SetAlFnt{\small}
\SetAlCapFnt{\small}
\SetAlCapNameFnt{\small}
\SetAlCapHSkip{0pt}
\IncMargin{-\parindent}

\author[1]{Weiran Shen}
\author[2]{Zihe Wang}
\author[3]{Song Zuo}

\affil[1]{IIIS, Tsinghua University\protect \\ \texttt{emersonswr@gmail.com}}
\affil[2]{ITCS, Shanghai University of Finance and Economics\protect \\ \texttt{wang.zihe@mail.shufe.edu.cn}}
\affil[3]{Google Research\protect \\ \texttt{szuo@google.com}}

\title{Bayesian Nash Equilibrium in First-Price Auction with Discrete Value Distributions\thanks{The codes for solving the BNE in our setting can be found at \url{https://github.com/emersonswr/discrete_fpa_bne}}}

\begin{document}

\maketitle

\begin{abstract}
First price auctions are widely used in government contracts and industrial auctions.
In this paper, we consider the Bayesian Nash Equilibrium (BNE) in first price auctions with discrete value distributions.
We study the characterization of the BNE in the first price auction and provide an algorithm to compute the BNE at the same time. Moreover, we prove the existence and the uniqueness of the BNE.
Some of the previous results in the case of continuous value distributions do not apply to the case of discrete value distributions. In the meanwhile, the uniqueness result in discrete case cannot be implied by the uniqueness property in the continuous case. 

Unlike in the continuous case, we do not need to solve ordinary differential equations and thus do not suffer from the solution errors therein.
Compared to the method of using continuous distributions to approximate discrete ones, our experiments show that our algorithm is both faster and more accurate.

The results in this paper are derived in the asymmetric independent private values model, which assumes that the buyers' value distributions are common knowledge.
\end{abstract}

\newpage

\section{Introduction}

Recently, the display
advertising industry has switched from second-price to
first-price auctions,\footnote{The winners now {\em pay the bid price} in the
	{\em Exchange Bidding} auctions by Google Ad Manager (Google's display
	ad-serving system for publishers). In Exchange Bidding auctions, the exchanges
	bid to compete for ad requests from publishers
	\url{https://support.google.com/admanager/answer/7128958}.} and one important
reason is that some advertisers no longer trust the exchange to honestly
follow the second-price auction rules \cite{sluis2017big}. 
From a bidder's point of view, he does not trust the auctioneer since the auctioneer could also benefit from manipulating the auction rules after observing the sealed bids \cite{akbarpour2018credible}.
\citet{akbarpour2018credible} show that the
first-price auction is the unique \emph{credible} and \emph{static} optimal auction, which may be one potential backing theory for the trend of adopting the first-price auction
  in the ad exchange industry. \citet{akbarpour2018credible} also prove that no
  mechanism is static, credible, and strategy-proof (incentive compatible)
  at the same time. In particular, being credible means that it is incentive compatible for the
  auctioneer to follow the rules and being static roughly means that the auction is
  sealed-bid. Therefore, the first-price auction naturally becomes the only
  choice for the ad exchange industry, in which credibility becomes a major
  concern and sealed-bid is also critical to keep the
  auction process time-efficient for production needs.

  In contrast to the crucial needs from practice, the understanding of the
  first-price auction from auction theory remains shallow compared with that of
  other direct auctions. The essential obstacle is the complex equilibrium
  structure in first-price auctions. Following the first step by
  \citet{vickrey1961counterspeculation} for the symmetric setting, it has been a
  tough and long journey towards the existence, the uniqueness and the computation of the Bayesian Nash equilibrium of first-price auctions in general settings.\citet{plum1992characterization} covers the power distribution $F_1(x)=x^\mu$ and $F_2(x)=(\frac{x}{\beta})^\mu$ with the same support. \citet{kaplan2012asymmetric} solve the problem for uniform distributions with different support.

  In this paper, with the application in ad auctions as one of the important
  motivations, we focus on the computation of BNEs in first-price auctions 
  where the bidders' values are independently drawn
  from discrete prior distributions.
We study the discrete value setting for several reasons. First, it is a basic setting with a different structure from the continuous case. The results in the discrete setting can provide us with more structural insights that can not be obtained from the continuous case. Second, such a setting is more realistic in practice. Consider the situation where prior information is expensive or impossible to acquire. In such situations, we might only have historical samples of the buyer's values, and it is reasonable to aggregate these discrete values to form a discrete empirical distribution.

%

Our algorithm does not involve ordinary differential equations, which makes our algorithm robust and much faster.


\subsection{Our Contributions}
	\begin{itemize}
		\item
 	We give an efficient algorithm to find the BNE of the first price auction. For any possible bid, by scrutinizing the bidders who might report it, we give a clear characterization of the BNE in the discrete setting. 
Previous methods make use of Nash's Theorem to prove the existence of the equilibrium in the continuous case, while we provide a constructively proof in the discrete case.

	\item We show that the equilibrium is unique in the discrete case (Theorem \ref{thm:unique}).
	The uniqueness result by \citet{lebrun2006uniqueness} relies on a technical assumption about buyers' value distributions. In contrast, we do not need any assumption. Furthermore, in the continuous case, we need to be very careful when a buyer's value is near the smallest value.
	In the discrete case, each buyer's strategy around the smallest value has relatively nice properties.
	\end{itemize}

\subsection{Related Works}
\label{sec:related_works}
	Besides the closed-form solution of the equilibrium, there is also a line of papers that focus on other aspects of the problem  \cite{shubik1970different,riley1981optimal,milgrom1982theory}.
The existence of a Bayesian Nash equilibrium is given by \citet{lebrun1996existence,maskin2000equilibrium,athey2001single}. They first show the existence for discrete distributions by applying Nash's Theorem. Then they show the existence in the continuous case by constructing a series of discrete distributions that approaches the actual continuous case. In this paper, we prove the existence result by construction.

	After proving the existence of the BNE, researchers began to consider the its uniqueness.
	For symmetric distributions, \citet{chawla2013auctions} prove the uniqueness by ruling out asymmetric equilibria.
	For asymmetric distributions,
	\citet{maskin2003uniqueness} show that the equilibrium is unique for symmetric distributions with the assumption that there is positive mass at the lower possible value.
	\citet{lebrun1999first,lebrun2006uniqueness} prove the uniqueness for more general settings but still with the assumption that the cumulative value distribution functions are strictly log-concave at certain points.  \citet{escamocher2009existence} investigates the existence and the computation of BNEs in the discrete case, under the assumption that bidders can only place discrete bids. They consider both the randomized tie-breaking and the Vickery tie-breaking and give different results.
    
    However, both the continuous and the discrete case without assumptions are still left open. In this paper, we solve the discrete value distribution case.

	In the numerical analysis literature, \citet{marshall1994numerical} give the first numerical analysis for two special distributions. Their backward-shooting method then become the standard method for computing the equilibrium strategies of asymmetric first-price auctions \cite{bajari2001comparing,fibich2003asymmetric,li2007auction}. The backward-shooting method first computes the smallest winning bid, then repeatedly guess the largest winning bid and then solving ordinary differential equations all the way down in the bid space to see if the smallest winning bid given by the solution to the differential equations matches the actual one. 
	One common issue of this method is the computation error in solving ordinary differential equations. \citet{bajari2001comparing} uses a polynomial to approximate the inverse bidding strategy. To compute a solution with high precision, \citet{gayle2008numerical} use Taylor-series expansions. Our method belongs to the backward-shooting category. We do not need to solve ordinary differential equations, but the algorithm still needs to repeatedly guess the largest winning bid. \citet{fibich2012asymmetric} propose a forward-shooting method and numerically solve the case with power-law distributions. However this forward-shooting method does not work in the discrete case.


\section{Preliminaries}
	\label{sec:pre}
\subsection{Model}
	Suppose the seller has one item for sale and there are $n$ potential buyers $N=\{1,...,n\}$. The item is sold through a sealed-bid first-price auction. Each buyer has a private value for the item, which is drawn according to a publicly known value distribution. In our setting, we consider the case where the each buyer's value distribution is discrete. Also, we assume that for buyer $i$, the value support is a finite set $\{v_i^1,v_i^2,...,v_i^{d_i}\}$ with cumulative distribution function $G_i$, with $G_i(v)=\mathrm{Prob}\{v_i\le v\}$. Without loss of generality, we assume $0\le v_i^1<v_i^2<\ldots<v_i^{d_i}$. 
	
	Every buyer places a nonnegative bid $b_i$ simultaneously. Let $F_i(b_i)$ denote the cumulative distribution function of buyer $i$'s bids. We assume that buyers have quasi-linear utilities and no buyer overbids, i.e., no buyer will place a bid that is higher than his value. The buyer with the highest bid wins the item and pays what he bids. Each buyer's strategy is a mapping from his private value to his bid. The strategies form a Bayesian Nash Equilibrium (BNE) if no bidder has an incentive to change his strategy unilaterally in the Bayesian setting.	
	
	In the continuous value setting, each buyer's strategy maps a value to a bid. For example, suppose there are two i.i.d. buyers with value uniformly distributed between $[0,1]$. In the BNE, each buyer bids half of his private value.
	
	But in the discrete setting, each buyer's strategy is randomized, and maps a value to a set of possible bids, with a certain probability distribution. Consider the following example.
	\begin{example}
	\label{discrete2buyer}
	There are two i.i.d. buyers. Each buyer has value 1 and 2 with probability 0.5. 
	In the equilibrium, when a buyer's value is 2, it is possible for him to place any bid in $[1,1.5]$, and the bid density function is $\frac{1}{(2-x)^2}, \forall x\in[1,1.5]$.
	When a buyer's value is 1, the buyer bids 1 with probability 1.
	\end{example}
	\begin{figure}[h]
		\center
		\includegraphics[width=0.5\linewidth]{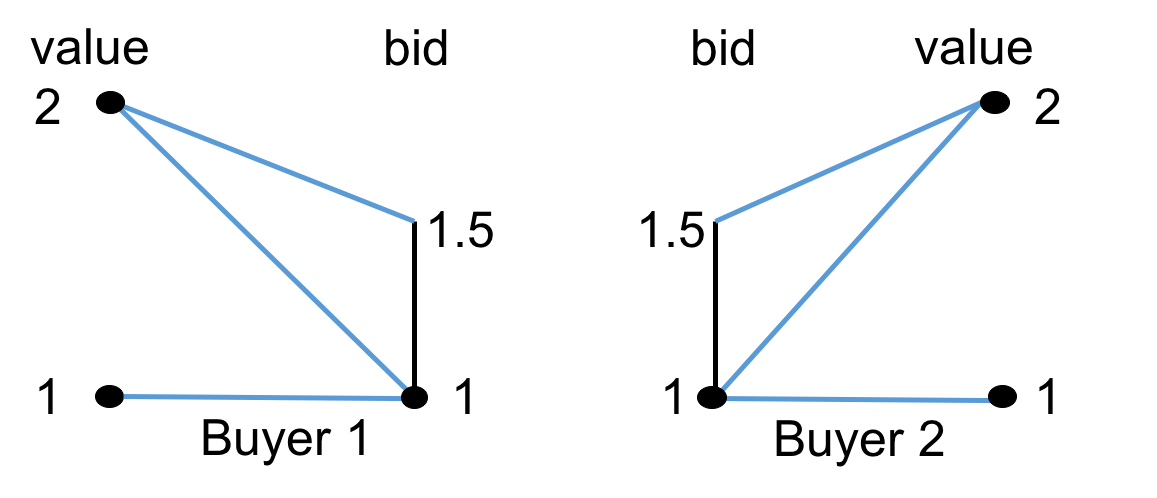}
		\caption{The equilibrium strategy of two i.i.d. buyers with uniform $\{1,2\}$ value distribution. Although the value distribution is discrete, the bids are continuous.}
		\label{diseg}
	\end{figure}

	Our objective is to find the bidding strategies that constitute a Bayesian Nash equilibrium. 	
	Before we start, we need to make an assumption of the tie-breaking rule to guarantee the existence of an equilibrium. To see this, suppose there are multiple highest bids, and we allocate the item randomly among the corresponding buyers. In the following example, there is no BNE.
	\begin{example}[\citet{maskin2000equilibrium}]
	There are two buyers. Buyer 1 has two possible values $\{1, 2\}$, both with probability 0.5, and Buyer 2 has only one possible value 1, with probability 1.

Suppose there is a BNE. There are 2 possible cases in the equilibrium.
  \begin{itemize}
  \item Buyer 2 always bids 1. Buyer 1 has value 2, he wins with probability 1 if he bids $1+\epsilon$ where $\epsilon$ is a small positive number. Buyer 1's utility is $2-1-\epsilon$. Thus buyer 1 always prefers a smaller $\epsilon$ as long as $\epsilon>0$. However if he chooses $\epsilon$ to be zero,
  Buyer 1 wins with probability 0.5 and the utility drops to $(2-1)\times 0.5=0.5$.   Thus Buyer 1 with value 2 has no best response.
  \item Buyer 2 bids 1 with probability strictly less than 1. Consider the lowest possible bid $\underline{b}$ that any buyer would place. Since no buyer overbids, we have $\underline{b}< 1$. If both buyers bid $\underline{b}$ with positive probability then they both have incentives to bid $\underline{b}+\epsilon$. It cannot be a BNE.
  If one buyer bids $\underline{b}$ with zero probability then the other buyer will not bid $\underline{b}$ because he will lose. It cannot be a BNE.
  \end{itemize}
  \end{example}

It turns out that there is no BNE in this example if we break ties randomly.
We consider an alternative tie-breaking rule introduced by \citet{maskin2000equilibrium}. When there are multiple highest bids, we will allocate the item to the buyer with the highest value.
	\begin{assumption}[\citet{maskin2000equilibrium}]
		\label{assumption}
	Ties are broken by running a Vickrey auction among the highest buyers.\footnote{This can be implemented by letting the highest buyers submit new bids.}
	\end{assumption}

	In the continuous value setting, this assumption is unnecessary, but in the discrete value setting, we need this assumption to deal with best response issue. 
	Without this assumption, we will still get an approximate BNE using our algorithm.
	We will discuss more about this assumption(see Example \ref{classic}).

\subsection{Basic structure of the BNE}
	To assist later arguments, we restate several properties of the Bayesian Nash equilibrium, mainly summarized by \citet{maskin2003uniqueness}.\footnote{Although they assume twice continuously differentiable value distributions and the buyers have the same upper limit of values, the lemmas still hold for the present setting.
} 
Giving all buyers' strategies, if any buyer can win with a certain probability by bidding $b$ then we call $b$ a \emph{winning bid}.
Without loss of generality, we assume the set of winning bids is closed and there exists a smallest winning bid $\underline{b}$. 
As long as a buyer's bid is higher than or equal to $\underline{b}$, he can win with a certain probability. 
If buyers' strategies form a BNE, the smallest winning bid $\underline{b}$ will be determined uniquely by the buyers' value distributions:
\begin{lemma}[\citet{maskin2003uniqueness}]
    \label{lem:min_winning_bid}
    Assume buyer $i^*$ has the largest smallest value, i.e., $v_{i^*}^1=\max_j v_j^1$. Then the smallest winning bid is
    \begin{gather*}
    \underline{b}= \argmax_b (v_{i^*}^1-b)\prod_{i\ne i^*} G_i(b).
    \end{gather*}
\end{lemma}
%
%

We claim that smallest bid placed by buyer $i^*$ is exactly $\underline{b}$:
\begin{itemize}
    \item On the one hand, $\underline{b}$ must be larger than or equal to the smallest bid of buyer 1. Otherwise, other buyers will have no incentive to place bid $\underline{b}$, since giving $\underline{b}$ always loses.
    \item On the other hand, $\underline{b}$ must be smaller than or equal to the smallest bid of buyer $i^*$, otherwise he will have negative utilities.
\end{itemize}  

According to Example \ref{discrete2buyer}, a buyer with a certain value $v^j_i$ may place multiple bids.
Let $S_i(v^j_i)$ be the set of possible bids for buyer $i$ when he has value $v^j_i$. For ease of presentation, we assume that $S_i(v^j_i)$ is a closed set, otherwise, we can take the closure of the support as $S_i(v^j_i)$. Denote buyer $i$'s all possible bids by $S_i$, i.e., $S_i=\bigcup_j S_i(v^j_i)$.

\citet{maskin2003uniqueness} show that for any winning bid $b$, there are at least two buyers $i$ and $j$, such that $b\in S_i$ and $b\in S_j$. 
The intuition is that any buyer who places a winning bid $b$ needs a competitor, otherwise, the buyer can place $b-\epsilon$ to increase his utility.
	\begin{lemma}[\citet{maskin2003uniqueness}]
		\label{lemma4}
		In the BNE of the first-price auction, for buyer $i$ and any $b_i>\underline{b}$, if $b_i\in S_i$, then there must exist another buyer, who bids in $(b_i-\epsilon, b_i)$ with positive probability for any $\epsilon$.
	\end{lemma}
\begin{proof}
Buyer $i$'s utility is positive since he can bid $\underline{b}+\epsilon$ which implies positive winning probability and positive utility conditioned on winning.
There is no buyer who bids $b_i$ with positive probability, otherwise buyer $i$ can increase his utility by bidding $b_i+\epsilon$.
If there is another buyer who bids in $(b_i-\epsilon, b_i)$ with zero probability, then buyer $i$ has an incentive to decrease his bid to $b_i-\epsilon$ which strictly increases his utility.
\end{proof}
	
	\citet{maskin2003uniqueness} show that, in first price auctions, a buyer would not give a particular bid with positive probability when this bid is larger than or equal to $\underline{b}$.
	\begin{lemma}[\citet{maskin2003uniqueness}]
	\label{nomass} 
		For any buyer $i$, there is no mass point above $\underline{b}$ in buyer $i$'s bid distribution.
	\end{lemma}

	The following lemma shows that when a buyer's value is larger than or equal to $\underline{b}$, his bidding strategy is monotone in his value.
	\begin{lemma}[\citet{maskin2003uniqueness}]
		For each buyer, his bidding strategy is monotone in value, i.e., $\max S_i(v^j_i)\leq \min S_i(v^{j+1}_i)$ for $v^j_i\geq \underline{b}$.
		\label{1pmonotone}
	\end{lemma}

Here we provide an example of what the BNE looks like in the discrete case. The computation of such a BNE will be clear after the analysis of our algorithm.
	\begin{example}
    \label{classic}
    Suppose there are 4 buyers with the following discrete value distributions:
    \begin{gather*}
    G_1(x)=\begin{cases}
    1& x=20\\
    \frac{11\sqrt{7}}{24\sqrt{3}}& x=10\\
    \frac{\sqrt{77}}{12\sqrt{2}}& x=2\\
    \end{cases},
    G_2(x)=\begin{cases}
    1& x=14\\
    \frac{4}{\sqrt{21}}& x=13\\
    \frac{2\sqrt{22}}{7\sqrt{7}}& x=1\\
    \end{cases},
    \end{gather*}
    
    \begin{gather*}
    G_3(x)=\begin{cases}
    1& x=20\\
    \frac{11}{12}& x=9\\
    \end{cases},
    G_4(x)=\begin{cases}
    1& x=12\\
    \frac{3\sqrt{3}}{2\sqrt{7}}& x=1\\
    \end{cases}.
    \end{gather*}
    
    In the BNE, the buyers bid according to the following bid distributions:
    
    \begin{gather*}
    F_1(x)=\begin{cases}
    \frac{11}{20-x}& x\in (8,9]\\
    \frac{11}{12}\sqrt{\frac{2(20-x)}{(14-x)(12-x)}}& x\in (6,8]\\
    \frac{77}{48}\sqrt{\frac{10-x}{(9-x)(13-x)}}& x\in [2,6]\\
    \end{cases},
    F_2(x)=\begin{cases}
    \sqrt{\frac{8(14-x)}{(20-x)(12-x)}}& x\in (6,8]\\
    \frac{8}{7}\sqrt{\frac{13-x}{(10-x)(9-x)}}& x\in [2,6]\\
    \frac{2\sqrt{22}}{7\sqrt{7}}& x=1\\
    \end{cases},
    \end{gather*}
    
    \begin{gather*}
    F_3(x)=\begin{cases}
    \frac{11}{20-x}& x\in [8,9]\\
    \frac{11}{6}\sqrt{\frac{7(9-x)}{3(10-x)(13-x)}}& x\in [2,6]\\
    \end{cases},
    F_4(x)=\begin{cases}
    \sqrt{\frac{18(12-x)}{(20-x)(14-x)}}& x\in [6,8]\\
    \frac{3\sqrt{3}}{2\sqrt{7}}& x=1\\
    \end{cases}.
    \end{gather*}

    \begin{figure}[h]
        \centering
        \includegraphics[width=0.5\linewidth]{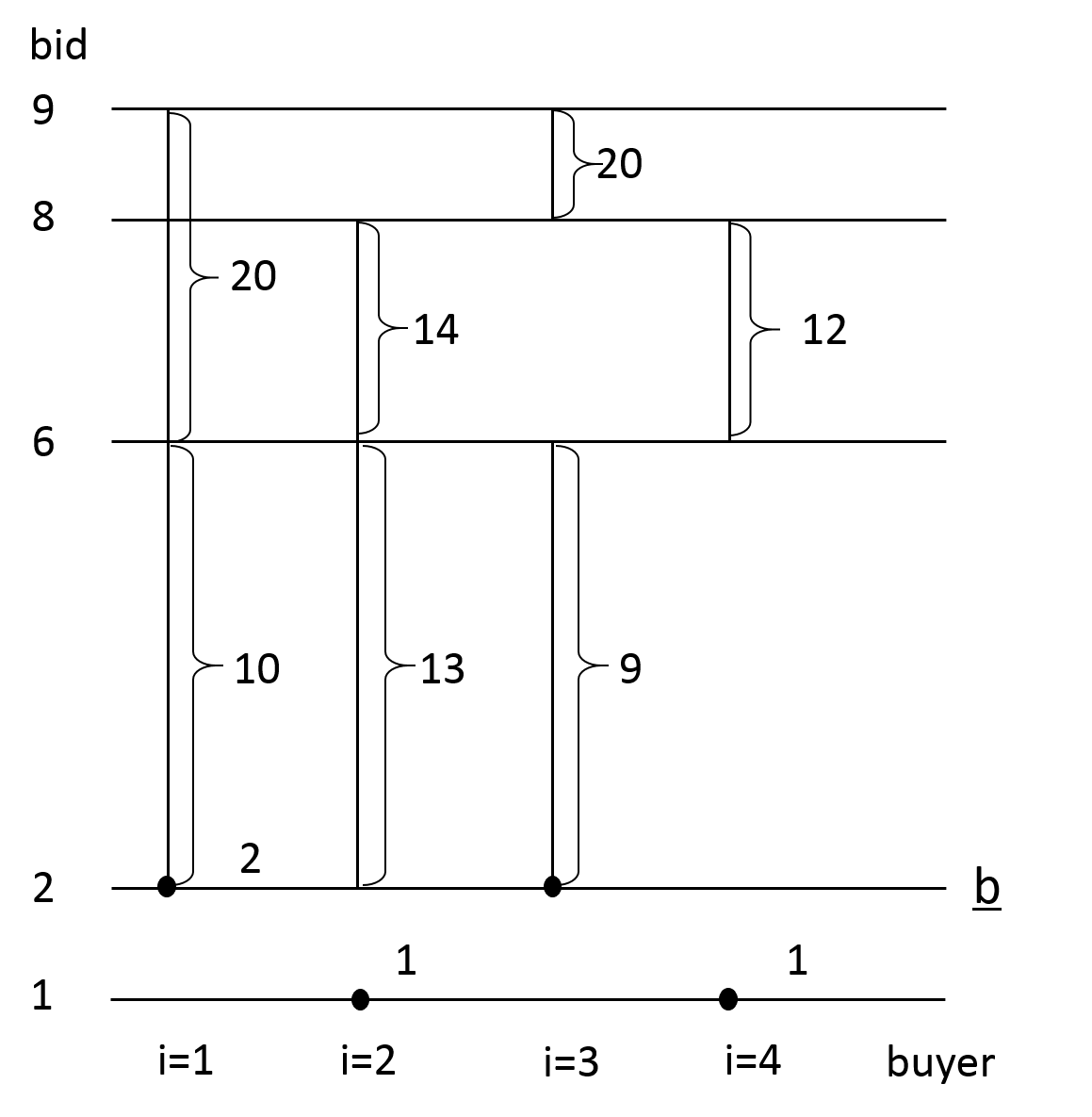}
        \caption{Each value's corresponding bid interval is indicated by braces. A dot implies a positive probability (point mass). Buyer 1 with value 2 and Buyer 3 with value 9 are both likely to bid 2. According to Assumption \ref{assumption}, Buyer 3 wins in this case.}
        \label{fig001}
    \end{figure}
    
\end{example}

\section{Overview of our algorithm}
As mentioned in Section \ref{sec:related_works}, our algorithm falls in the backward-shooting category. In first-price auctions, if the largest winning bid $\bar{b}$ is given, our theoretical analysis enables us to compute the buyers' bid distributions all the way down to the smallest winning bid $\underline{b}$. However, it turns out that we only know how to compute the smallest winning bid $\underline{b}$ (Theorem \ref{lem:min_winning_bid}), but have no idea about $\bar{b}$. Therefore, the backward-shooting algorithms just repeatedly guess $\bar{b}$ and see if the computed $\underline{b}(\bar{b})$ matches the actual $\underline{b}$, reducing the computation of the BNE to a searching problem. An overview of the backward-shooting algorithm is shown in Algorithm \ref{alg:overview}, where a binary search algorithm is used.
\begin{algorithm}[h!]
    \SetAlgoLined
    \SetKwInOut{Input}{Input}
    \SetKwInOut{Output}{Output}
    
    \Input{Buyers' value distributions $G_i$.}
    \Output{Buyers' bid distributions $F_i$.}
    
    Compute the smallest winning bid $\underline{b}$ using Lemma \ref{lem:min_winning_bid} \;
    $UB\gets\max\{\cup_{i\in N} ~\mathrm{supp}(G_i)\}$, $LB\gets0$ \;
    \While{some exit condition is not met}{
        $\bar{b}\gets \frac{1}{2}(UB+LB)$\;
        Compute $F_i$ all the way down from $\bar{b}$ to the corresponding smallest winning bid $\underline{b}(\bar{b})$\;
        \eIf{$\underline{b}(\bar{b})>\underline{b}$}{$UB\gets\bar{b}$ \;}{$LB\gets \bar{b}$ \;}
    }
    \Return  $F_i$\;
    \caption{Overview of the backward-shooting algorithm}
    \label{alg:overview}
\end{algorithm}

In Algorithm \ref{alg:overview}, $UB$ and $LB$ are the upper bound and the lower bound of the largest winning bid. According to Lemma \ref{lem:min_winning_bid}, the smallest winning bid $\underline{b}$ can be easily determined. The exit condition measures how close our guess of $\bar{b}$ is to the actual largest winning bid, for example, we can check if $UB-LB<\epsilon$ or we can compare whether $\underline{b}(\bar{b})$ is close enough to $\underline{b}$. And guaranteed by Theorem \ref{thm:mono}, \ref{continuity} and Corollary \ref{strict}, we know that $\underline{b}(\bar{b})$ is monotone with respect to the current guess $\bar{b}$, and how to adjust $UB$ and $LB$ accordingly.

\begin{algorithm}
    \SetAlgoLined
    \SetKwInOut{Input}{Input}
    \SetKwInOut{Output}{Output}
    
    \Input{the largest winning bid guess $\bar{b}$.}
    \Output{the smallest winning bid $\underline{b}(\bar{b})$.}
    
    Initialize $b \leftarrow \bar{b}, \Lambda(b) \leftarrow \emptyset$ \;
    Update $\Lambda(b)$ by repeatedly adding buyers to $\Lambda(b)$ according to Theorem \ref{thmenter}\;
    \While{$|\Lambda(b)|\geq 2$ and $b> 0$}{
        Predict the next change position $b'$ according to Theorem \ref{thmrule1leave} and
        \ref{thmenter}\;
        Set $b\gets b'$\;
        Update the bidding set $\Lambda(b')$ by removing buyers  according to Theorem \ref{thmrule1leave}\;
        Update the bidding set $\Lambda(b')$ by adding buyers  according to Theorem \ref{thmenter}\;
    }
    \Return  $\underline{b}(\bar{b})=b$\;
    \caption{Compute $\underline{b}(\bar{b})$ given guess $\bar{b}$.}
    \label{alg}
\end{algorithm}

In the continuous value distribution case, the computation of $\underline{b}(\bar{b})$ given $\bar{b}$ is done through solving ordinary differential equations. However, in the discrete case, we compute $\underline{b}(\bar{b})$ with Algorithm \ref{alg}. We define a core structure called the ``bidding set'' (Definition \ref{def:bidding_set}), and starting from $\bar{b}$, we update the bidding set as the bid goes down (i.e., compute where each buyer enters or leaves the bidding set with all his values). Each buyer enters the bidding set with his largest ``unconsumed'' value when certain conditions are met (Theorem \ref{thmenter}), and leaves the bidding set when the probability of the corresponding value is used up (Theorem \ref{thmrule1leave}), hence the value is ``consumed''. After all buyers consumed all their values, the corresponding $\underline{b}(\bar{b})$ is found. In our setting, Algorithm \ref{alg} is used in Algorithm \ref{alg:overview} as a sub-routine.

\section{The bidding set}
\label{sec:bidding_set}
Our objective is to compute every buyer's strategy in BNE.
Since a buyer's bidding strategy is monotone, it suffices to compute the bid distribution because we can map a buyer's value to a bid with the same quantile in his bid distribution. However, when the bid support is not continuous, there might exist multiple bids with the same quantile. 

In this section, we analyze the structure of the BNE. We introduce a useful tool called the ``bidding set'', and analyze how the bidding set changes in the bid space. With the analysis of the structure, an algorithm of computing the BNE can be naturally derived.
	\begin{definition}[Bidding set and waiting list]
		\label{def:bidding_set}
	The set of buyers whose bidding strategies include bid $x$ is called the \emph{bidding set}, denoted by $\Lambda(x)$, i.e., $\Lambda(x)=\{i\mid x\in S_i \}$.
    
    The set of other buyers $N-\Lambda(x)$ is called the \emph{waiting list} at bid $x$.
\end{definition}	
When there is no ambiguity, we use $\Lambda$ instead. The following theorem is about the relationship between the bid  distribution and $\Lambda$.
For any buyer set $\Lambda$, let $F_\Lambda(x)$ denote the product of the cumulative bid distribution of buyers in $\Lambda$, i.e., $F_\Lambda(x)=\prod_{i\in \Lambda(x)} F_i(x)$.
We abuse notation and use $v_i(x)$ to represent player $i$'s value when he bids $x$ in the equilibrium ($v_i(x)$ is well-defined according to Lemma \ref{1pmonotone}). 

Before discussing how the bidding set $\Lambda$ changes in the bid space, let's first consider some properties of the bidding set.

Define function
\begin{gather*}
h_i(x)=\frac{1}{|\Lambda(x)|-1}\left(\sum_{j\in \Lambda(x)}\frac{1}{v_j-x}\right)-\frac{1}{v_i-x}, \forall i\in \Lambda(x).
\end{gather*}
	\begin{theorem}
	Suppose $\Lambda(x)$ does not change in bid interval $(b_1,b_2)$, and $v_i(x)$ is constant for $x\in(b_1,b_2), i\in \Lambda(x)$. Then the bid distribution of every buyer in $\Lambda$ is differentiable  in this interval. In fact, for any $x\in (b_1,b_2)$ we have
	$$ \frac{f_i(x)}{F_i(x)}=h_i(x), \forall i \in \Lambda.$$
	
	\label{thmeqa}
	\end{theorem}

	\begin{proof}
	For any bid $x\in(b_1,b_2)$, the utility of buyer $i$ is:
	\begin{eqnarray}
	u_i(v_i)=(v_i-x)F_{N\setminus \Lambda}(x)F_{\Lambda\setminus i}(x), \forall i\in \Lambda\label{eq1}
	\end{eqnarray}
	Since buyers in $N\setminus\Lambda$ do not bid in interval $(b_1,b_2)$, we have $F_{N\setminus \Lambda}(x)=F_{N\setminus \Lambda}(b_1)$. 
	After multiplication over $i\in \Lambda$, we have
	$$\prod_{i\in \Lambda} u_i(v_i)=\prod_{i\in \Lambda}(v_i-x)(F_{\Lambda}(x))^{|\Lambda|-1}(F_{N\setminus \Lambda}(b_1))^{|\Lambda|},$$
	where $|\Lambda|$ is the number of buyers in the bidding set $\Lambda$.
	We use Equation \eqref{eq1} to cancel out the term $F_\Lambda(b_1)$ and get 
	$$F_i(x)=\frac{v_i-x}{u_i(v_i)}\left(\prod_{i\in \Lambda} \frac{u_i(v_i)}{v_i-x}\right)^{\frac{1}{|\Lambda|-1}} \left(F_{N\setminus \Lambda}(x)\right)^{-\frac{1}{|\Lambda|-1}}$$
	Since the right hand side of the equation is differentiable with $x$, the left hand side $F_i(x)$ is differentiable.
	Take derivatives on both sides, we get 
    \begin{align*}
    f_i(x)=&\left[\frac{1}{|\Lambda|-1}\left(\sum_{j\in \Lambda}\frac{1}{v_j-x}\right)-\frac{1}{v_i-x}\right]F_i(x)\\
    =&h_i(x)F_i(x).
    \end{align*}
	\end{proof}

   
If we know what the bidding set $\Lambda(x)$ is for every possible $x$ in the BNE, we can construct each buyer's bid distribution $f_i(x)$ according to Theorem \ref{thmeqa}. Therefore, the rest of this section is devoted to the analysis of how the bidding set changes.
   
    Now we discuss the basic component $S_i(v^j_i)$.
	When buyers' values are continuous, it is known that the support of the bidding strategy is connected for every buyer\cite{lebrun1999first}.
	This result no longer holds in the discrete value setting. In Example \ref{classic}, Buyer 3's possible bids have two connected parts. However, we can prove a weaker version of this structural result.
	\begin{theorem}
		$S_i(v^j_i)$ is an interval when $v^j_i\geq \underline{b}$.
		\label{thmbidcont}
	\end{theorem}

We defer the proof to Appendix \ref{appendix:interval_proof}.

\begin{remark}
    In Example \ref{classic},  ties happen with zero probability when the winning bid is greater than $\underline{b}$, but with strict positive probability when the winning bid is $\underline{b}$.
Buyer 1 with value 2 bids 2 and Buyer 3 with value 3 is also likely to bid 2. 
By Assumption \ref{assumption}, Buyer 3 is the winner since he has a greater value. 
Actually, Assumption \ref{assumption} is only used in dealing with the best response issue when players tie at the smallest winning bid $\underline{b}$. Suppose the unique BNE under this assumption is $\mathcal{E}$. We can create an approximate BNE with $\epsilon$-BNE $\mathcal{E}'$ which does not need the tie-breaking assumption, i.e., the highest bid are resolved randomly. The only change is that when a player is supposed to bid $\underline{b}$ in $\mathcal{E}$ and this player’s value is larger than $\underline{b}$, we let this player bid $\underline{b}+\epsilon$ in $\mathcal{E}'$ instead. Therefore, the tie-breaking rule can be viewed as a way of obtaining an approximate BNE. Once the output of the algorithm is obtained, it is still an approximate BNE even without the tie-breaking rule.
\end{remark}

%
%

\subsection{Change points of the bidding set}
	\label{sec:cha2}
	In this section, we consider some properties of the bidding set at its change points. These results are helpful for computing these change points.
	
	\begin{definition}
When bidding set changes at x, we use $\Lambda^+(x)$ and $\Lambda^-(x)$ to denote the buyers who bid in the upper neighborhood and lower neighborhood around $x$, i.e.,
		$$\Lambda^+(x)=\left\{i \mid \exists \epsilon>0, (x,x+\epsilon)\subseteq S_i \right\},$$
		$$\Lambda^-(x)=\left\{i \mid\exists \epsilon>0, (x-\epsilon,x)\subseteq S_i \right\}.$$
	\end{definition}

	\begin{example}
		In Figure \ref{fig001}, $\Lambda(8.5)=\{1,3\}$, $\Lambda^+(6)=\{1,2,4\}$ and $\Lambda^-(6)=\{1,2,3\}$. The waiting list at $x=8.5$ is $\{2,4\}$.
	\end{example}
	
	Clearly, when a bidding interval $S_i(v^j_i)$ starts or ends at a certain bid $x$, $\Lambda(x)$ changes. Therefore, to characterize how the bidding set changes, it suffices to determine when a bid interval $S_i(v^j_i)$ starts and ends, or equivalently, when a buyer enters the bidding set from the waiting list and vice verse. Our method falls in the backward-shooting category, thus we compute the buyers' bidding strategy from the largest winning bid all the way down.
	\begin{definition}
		We say a buyer \emph{enters} the bidding set with value $v_i$ at bid $x$ if $x=\max S_i(v_i)$. Similarly, we say a buyer \emph{leaves} the bidding set if $x=\min S_i(v_i)$.
	\end{definition}
	 \begin{remark}
	 	Notice that entering the bidding set with value $v_i$ is different from:
	\begin{gather*}
	i\not \in \Lambda^+(x),\text{ and }i\in \Lambda^-(x).
	\end{gather*}
	The reason is that it is possible for the buyer to leave the bidding set with another value $v'_i$ and enters immediately at the same point, but with a different value $v_i$. In this case, the buyer is always in the bidding set around $x$, but his value changes.
	 \end{remark}


%
	
\subsection{When to enter the bidding set}
\label{middle_change}

The following lemma gives an important property at the point where a buyer enters the bidding set.

	\begin{lemma}
	If buyer $i$ with value $v_i$ enters the bidding set at point $b$, then we have
    \begin{gather*}
    \frac{1}{|\Lambda^+(b)|-1}\sum_{j\in \Lambda^+(b)}\frac{1}{v_j-x}\leq \frac{1}{v_i-x}, \forall x\in (b,b+\epsilon),\\
    \frac{1}{|\Lambda^-(b)|-1}\sum_{j\in \Lambda^-(b)}\frac{1}{v_j-x}\geq \frac{1}{v_i-x}.\forall x\in (b-\epsilon,b).
    \end{gather*}
	\label{lemmachange}
	\end{lemma}
	\begin{proof}
    The first inequality is because buyer $i$ has no incentive to bid $b+\epsilon$ instead of $b$.
	The second inequality is because of Theorem \ref{thmeqa}. 
	\end{proof}
	To determine the exact entering point, we introduce $\phi^*(x)$.
	\begin{definition}
		Given the bid $x$, define virtual value $\phi^*(x)$ that satisfies the function:
		$$\frac{1}{|\Lambda(x)|-1}\sum_{i\in \Lambda(x)}\frac{1}{v_i-x}=\frac{1}{\phi^*(x)-x}$$
	\end{definition}
	The definition of $\phi^*(x)$ is based on $\Lambda(x)$. When $\Lambda(x)$ changes, $\phi^*(x)$ also changes as a consequence.
	
	\begin{theorem} \label{strict_decrease}
		$\phi^*(x)$ strictly decreases with respect to $x$.
	\end{theorem}	
	
The following theorem determines when should a buyer enter the bidding set.
\begin{theorem}
    \label{thmenter}
    Suppose buyer $i$ has the largest unconsumed value $v_i$ in the waiting list, he will enter the bidding set immediately when either one of the following two conditions is satisfied.
    \begin{itemize}
        \item $|\Lambda|\leq 1$ and $v_i>x$;
        \item $\frac{1}{v_i-x}\leq \frac{1}{|\Lambda|-1}\sum_{i\in \Lambda}\frac{1}{v_j(x)-x}$ (or equivalently $h_i(x)\ge 0$) and $v_i>x$.
    \end{itemize}
\end{theorem}
\begin{example}
    Consider bid 6 in Example \ref{classic} and Figure \ref{fig001}.
    Bidding set $\Lambda^+(6)$ is $\{1,2,4\}$, with corresponding values 20, 14, and 12. 
    all buyers in $\Lambda^+(6)$ have consumed the probability of their current value at bid 6, and
    they all leave the bidding set. Thus the bidding set becomes empty and the waiting list contains all buyers. Buyer 2 has the largest unconsumed value 13 in the waiting list.
    Since there is no buyer in the bidding set, according to Theorem \ref{thmenter},
    buyer 2 enters the bidding set.
    
    Next, buyer 1 has the largest unconsumed value in the waiting list.
    Since there is only one buyer (buyer 2) in the bidding set, according to Theorem \ref{thmenter},
    he also enters the bidding set.
    
    Then, buyer 3 and 4 are in the waiting list with values $9$ and $1$.
    Now the first condition in Theorem \ref{thmenter} cannot be satisfied since the current bidding set already contains two buyers.
    For buyer 3, we have $1/(9-6)\leq 1/(10-6)+1(13-6)$, satisfying the second condition.
    Therefore, buyer 3 also enters the bidding set.
    However, buyer 4 has value 1 which is smaller than the current bid, so buyer 4 is not eligible to enter the bidding set.
    After the update, we have $\Lambda^-(6)=\{1,2,3\}$.
\end{example}
	
    \subsection{When to exit the bidding set}
	The probability of a buyer's value being $v_i$ should equal the probability that he bids in the interval $S_i(v_i)$. By Theorem \ref{thmbidcont}, the bid set $S_i(v_i)$ of a specific value $v_i$ is a connected interval. Therefore, once we know the maximum bid in $S_i(v_i)$, we can compute his bid distribution all the way down (according to Theorem \ref{thmeqa}) until the bid distribution consumes all the corresponding value probability. 
	\begin{theorem}
	\label{thmrule1leave}
	Buyer $i$ with value $v_i^k$ leaves the bidding set at $x$ when the cumulative probability of bidding set equals to the probability of the value, i.e.,
	$F_i\left(\max S_i(v_i^k)\right) -F_i(x)=G_i(v_i^k)-G_i(v_i^{k-1})$.
	\end{theorem}
\begin{example}
		Consider Example \ref{classic} and Figure \ref{fig001}, $S_1(20)$ begins at bid 9 and consumes the probability of value 20 at bid 6. The probability that buyer 1 bids in $S_1(20)$ is 
	\begin{gather*}
	F_1(9)-F_1(6)=\frac{11}{20-9}-\frac{77}{48}\sqrt{\frac{10-6}{(9-6)(13-6)}}=1-\frac{11}{24}\sqrt{\frac{7}{3}},
	\end{gather*}
	which equals the probability of value 20.
\end{example}

\subsection{Monotonicity of entering and exiting points}

Now we present some monotonicity results in the discrete setting. These results are similar to the continuous case, but with different proofs.

For convenience, we define
\begin{gather*}
p_i^j=\ln G_i(v_i^{j})-\ln G_i(v_i^{j-1}), \forall i=1,\dots,n, j=2,\dots,i_k.
\end{gather*}
So $\{p_i^j\}_{i=1,\dots,n, j=2,\dots,i_k}$  uniquely determines the value distribution $G$. When there is no ambiguity, we use $\{p_i^j\}$ for simplicity.	We use $\mathcal{E}(\bar{b}, \{p_i^j\})$ to denote the set of bidding intervals given by Algorithm \ref{alg} with a guessed largest bid $\bar{b}$ and distribution $G$.

\begin{theorem}
    \label{thm:mono}
    The extreme points of every bid interval in $\mathcal{E}(\bar{b}, \{p_i^j\})$ is monotone in $\bar{b}$.
\end{theorem}
The proof is different from the continuous distribution. 
We prove it by analyzing the algorithm directly.

\begin{corollary}
    \label{strict}
    The position $\underline{b}(\bar{b})$ where Algorithm \ref{alg} stops is strictly monotone in $\bar{b}$.
\end{corollary}
It is possible that some bid intervals remain at the same positions. 
But the position where the algorithm stops increases strictly. 
Consider Example \ref{classic} and Figure \ref{comparet}, the end point in $\mathcal{E}(8.5, \{p_i^j\})$ is 0.35 and 
the end point in $\mathcal{E}(9.5, \{p_i^j\})$ is 3.76.

\begin{figure}[ht]
    \centering
    \includegraphics[width=0.7\linewidth]{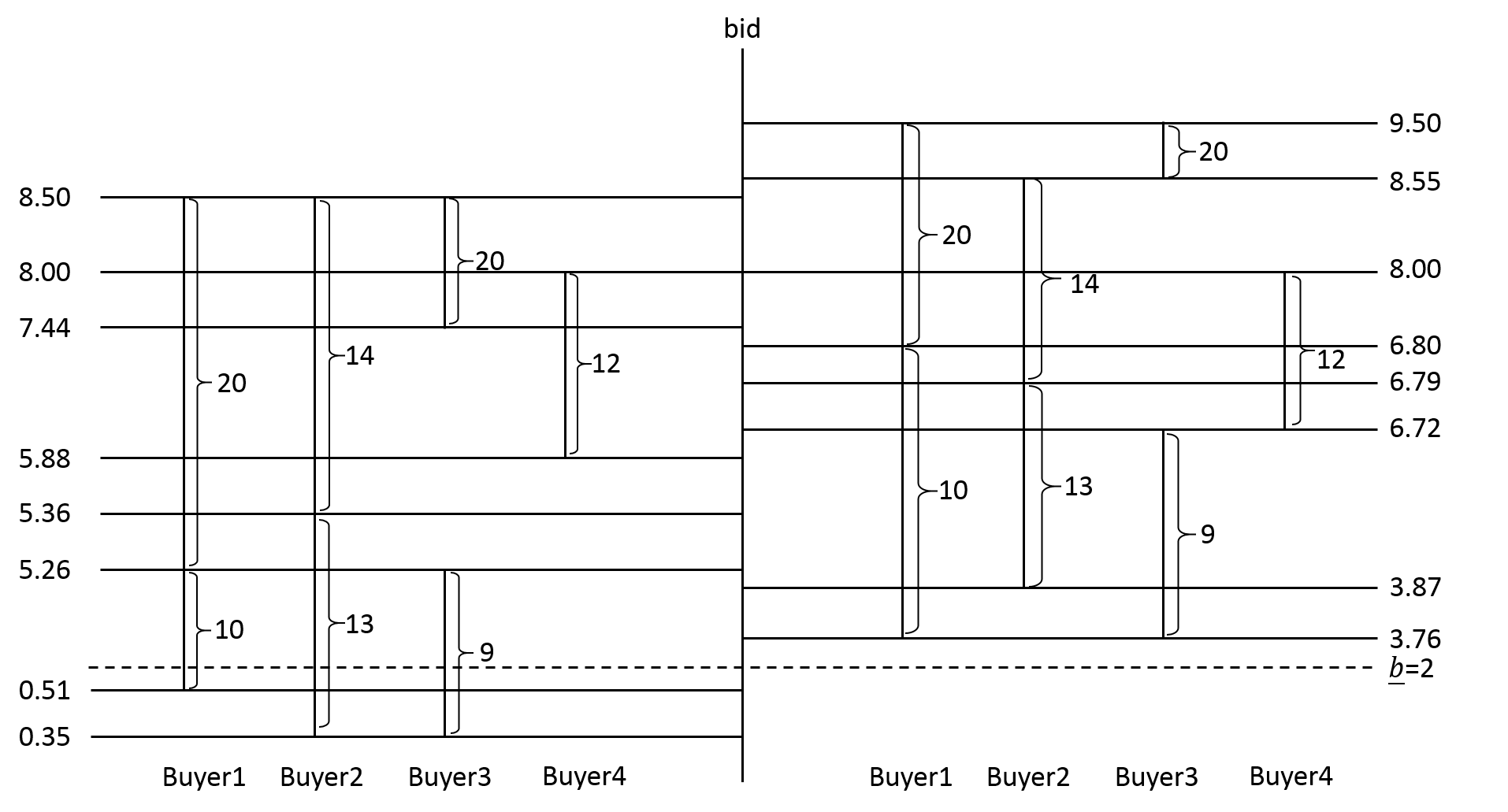}
    \caption{Monotonicity of $\underline{b}(\bar{b})$. In the left part, the guess is $\bar{b}=8.5$ and the algorithm stops at $\underline{b}(\bar{b})=0.35$. In the right part, the guess is $\bar{b}=9.5$, and the algorithm stops at $\underline{b}(\bar{b})=3.76$.}
    \label{comparet}
\end{figure}

Next we prove the continuity of the extreme points of each bid interval when the guess $\bar{b}$ changes.
\begin{theorem}
    The limit of each bid interval constructed by Algorithm \ref{alg} with the largest winning bid approaching to $\bar{b}^1$, is same as the bid interval constructed with the largest winning bid $\bar{b}^1$.
    \label{continuity}
\end{theorem}

\section{Existence and uniqueness of the BNE}
	\label{sec:exist}
		\subsection{Existence}
	In Algorithm \ref{alg}, if the point $\underline{b}(\bar{b})$ where the algorithm terminates does not match the actual smallest winning bid $\underline{b}$, 
	the bidding strategies we get do not form a BNE. But we show that if it does match $\underline{b}$, then the corresponding strategies do form a BNE.
%
	
	\begin{lemma}\label{lem:match}
		If $\underline{b}(\bar{b})$ matches the smallest winning bid $\underline{b}$, the bidding strategies given by Algorithm \ref{alg} is indeed a BNE.
	\end{lemma}
	
	\begin{theorem}
	A Bayesian Nash Equilibrium always exists when buyers have discrete value distributions.
	\end{theorem}
\begin{proof}
    If we guess $\bar{b}=\max{v^j_i}$, the $\underline{b}(\bar{b})$ given by Algorithm \ref{alg} would be $\max{v^j_i}$.
    And if we guess $\bar{b}=\min{v^j_i}$, the corresponding $\underline{b}(\bar{b})$ given by Algorithm \ref{alg} would be smaller than $\underline{b}$.
    By Theorem \ref{continuity}, there always exists a $\bar{b}^*$ such that the $\underline{b}(\bar{b}^*)=\underline{b}$. According to Lemma \ref{lem:match}, the guess $\bar{b}^*$ indeed gives a BNE.
\end{proof}

	Previous studies use Nash's Theorem to prove the existence of the BNE in the discrete value case. They first prove the existence in the setting where the bid space is restricted to a set of finite discrete bids. Then the existence in the continuous bid space is proved by taking the limit of the equilibrium in the restricted bid space setting. However, how to find a Bayesian Nash equilibrium is previously unknown.

%
 	\citet{maskin2003uniqueness} also propose a forward shooting method to compute the equilibrium. 
	They construct the BNE from the smallest winning bid. Under some conditions, they give the solution to the ordinary differential equations around the smallest winning bid.
	However, this method does not work for the discrete value setting. Because unlike the continuous value setting where the bidding set contains all buyers, in the discrete value setting, the bidding set at the smallest winning bid is unknown in the beginning.

	Following Algorithm \ref{alg}, it is easy to see that for two buyers with identical value distributions, they must enter and leave the bidding set exactly at the same points. 
\begin{corollary} 
	Buyers with identical value distributions have identical bidding strategies in the BNE. Furthermore, if all buyers have identical value distributions, i.e. symmetric distributions, the Bayesian Nash equilibrium is also symmetric.
	\end{corollary}
	
	
%
%

	\subsection{Uniqueness}
	According to the monotonicity of $\underline{b}(\bar{b})$, we can check whether the guessed bid is too high or too low. 
	Corollary \ref{strict} implies that only one guess of the largest winning bid $\bar{b}$ can possibly equate the corresponding $\underline{b}(\bar{b})$ and the actual smallest winning bid $\underline{b}$.
	
	\begin{theorem}\label{thm:unique}
	There exists a unique Bayesian Nash equilibrium when buyers have discrete value distributions.
	\end{theorem}
\begin{remark}
By uniqueness, we mean the equilibrium above the smallest winning bid is unique.
In Example \ref{classic}, we can change buyer 2's bids below 2 to any other bids below 2, and still get an equilibrium. But we only focus on the structure above the smallest winning bid.
\end{remark}	
	
In the continuous distribution case, in order to prove the uniqueness result, \citet{lebrun2006uniqueness} relies on the assumption that the distribution is strictly log-concave ($f_i/F_i$ is strictly decreasing) at the highest lower extremity of the supports, i.e., $v^1_1$ in our example. Briefly speaking, he uses this assumption to handle the case where some buyers always give bids that is larger than $\underline{b}$. However, in the discrete setting, according to our results, we do not need to deal with such cases. 
	


Furthermore, our algorithm does not have the sensitivity issue that are difficult to address in continuous algorithms.

\section{Experiments}
\subsection{Implementation}
\label{sec:det}

Algorithm \ref{alg:overview} computes the BNE of the first-price auction by repeatedly guessing the largest winning bid $\bar{b}$. According to Theorem \ref{thmeqa}, the change points of the bidding set completely determines the buyers' bidding strategies in the BNE.  And given a guess $\bar{b}$, we compute the change points of the bidding set $\Lambda(x)$ using Algorithm \ref{alg}.

Given a guess of the largest winning bid $\bar{b}$, we need to first determine the initial bidding set (line 2 in Algorithm \ref{alg}). This can be easily done by setting the initial bidding set $\Lambda(\bar{b})=\emptyset$, and then adding buyers to $\Lambda(\bar{b})$ according to Theorem \ref{thmenter}.




In Algorithm \ref{alg:overview}, different guesses of the highest winning bid may result in different bidding strategies of the buyers.  However, only the correct guess of the largest winning bid will produce the actual BNE. Since $\underline{b}(\bar{b})$ is monotone in $\bar{b}$, we can compare $\underline{b}(\bar{b})$ with the actual smallest winning bid given by Lemma \ref{lem:min_winning_bid} to check whether the guess is too high or too low, and shorten the interval where the actual highest winning bid can lie.

\subsubsection{Predicting the change points of the bidding set}
Given a guess of the largest winning bid $\bar{b}$, our algorithm computes the bidding strategies all the way down from $\bar{b}$ to $\underline{b}(\bar{b})$.
A naive way to compute how the bidding set changes is to decrease the bid gradually and see whether any buyer would leave or enter the bidding set. However, we can accelerate this process by predicting the next change point directly.

Suppose the current change point is $b$ and we want to predict the next change point $b'<b$. Assume that the bidding set changes at $b'$ because of buyer $i$\footnote{It is possible that the bidding set changes because of multiple buyers. However, the analysis still applies}. There are two possibilities:
\begin{itemize}
    \item buyer $i$ enters the bidding set with value $v_i^j$;
    \item buyer $i$ leaves the bidding set with value $v_i^j$.
\end{itemize}

For the first possibility, the two conditions in Theorem \ref{thmenter} are satisfied at point $b'$. Therefore, we must have
\begin{gather}
\label{eq:predict_enter}
v_i^j\ge b'\text{~and~}h_i(b')\ge 0.
\end{gather}

By assumption, the bidding set $\Lambda$ does not change in the interval $(b',b)$. Thus $h_i(x)$ can be easily computed in this interval and we can easily find the largest solution $b'$ to the above inequalities.

It is worth mentioning that it can be sometimes problematic when solving equation $h_i(x)=0$, since it involves division operations. However, this can be easily overcome by solving the following polynomial and discarding inappropriate solutions:
\begin{gather}
\label{eq:enter_poly_form}
0=\left[(v_i-x)\prod_{j\in \Lambda}(v_j-x) \right]h_i(x)
=\frac{v_i-x}{|\Lambda|-1}\sum_{j\in \Lambda}\prod_{k\in \Lambda\setminus\{j\}}(v_k-x)-\prod_{j\in \Lambda}(v_j-x).
\end{gather}



For the second possibility, we need to check when the probability of $v_i^j$ will be consumed completely. Let $\alpha$ be the remaining probability of $v_i^j$ at point $b$:
\begin{gather*}
\alpha=F_i(b)-G_i(v_i^{j-1}).
\end{gather*}
By Theorem~\ref{thmeqa}, we have
\begin{gather*}
\frac{\mathrm{d}\ln F_i(x)}{\mathrm{d}x}=h_i(x).
\end{gather*}
Define
\begin{gather*}
H_i(x)=\ln (v_i-x)-\frac{1}{|\Lambda(x)|-1}\sum_{j\in \Lambda}\ln (v_i-x).
\end{gather*}

It's easy to see that $H'_i(x)=h_i(x)$. Therefore, to compute the leaving point of buyer $i$, we need to solve
\begin{gather*}
\int_{b'}^{b}\,\mathrm{d}\ln F_i(x)=\int_{b'}^{b}h_i(x)\,\mathrm{d}x=\int_{b'}^{b}\,\mathrm{d}\ln H_i(x).
\end{gather*}
Or equivalently,
\begin{gather}
\ln F_i(b)-\ln (F_i(b)-\alpha)=\ln F_i(b)-\ln G_i(v_i^{j-1})=H_i(b)-H_i(b').
\label{eq:predict_leave}
\end{gather}

In fact, every buyer could cause the bidding set to change. Thus, to compute the actual next change point, we need to enumerate all possibilities, i.e., for each buyer $i$, assume the bidding set changes due to this buyer and compute the change point $b'_i$. And the actual change point is the largest among them:
\begin{gather*}
b'=\max_i\{b'_i \}.
\end{gather*}


\subsubsection{Updating the bidding set}
After predicting the next change point $b'$ and the corresponding buyer $i$, we need to update the bidding set $\Lambda(b')$. Of course, buyer $i$ should be removed from or added to the bidding set accordingly. But it is also possible that other buyer may enter or leave the bidding set as well, since the bidding set has changed.

According to Theorem \ref{thmrule1leave}, if a buyer is already in the bidding set, he will only leave when the corresponding value has been consumed. Therefore, at point $b'$, any buyer in $\Lambda$ with unconsumed value will not leave $\Lambda$ no matter how $\Lambda$ changes. So if we remove all buyers whose values are consumed at $b'$, we can guarantee that all the remaining buyers in $\Lambda$ will still be in $\Lambda$ after the update.

The next step is to add the buyers in the waiting list to the bidding set. Suppose buyer $j$ is the buyer with the largest unconsumed value in the waiting list. According to Theorem \ref{thmenter}, we can check whether buyer $j$ will enter the bidding set or not. There are two cases:
\begin{enumerate}
    \item If buyer $j$ enters the bidding set, then we have new buyer with the largest unconsumed value, and we can repeat this process until the buyer with the largest unconsumed value does not enter the bidding set, which becomes the next case;
    \item If buyer $j$ does not enter the bidding set, then all the other buyers in the waiting list does not enter, either. This can be easily proved using Theorem \ref{thmenter} by contradiction.
\end{enumerate}

Therefore, in Algorithm \ref{alg}, when updating the bidding set, we first remove buyers from it and then add buyers from the waiting list.

\subsubsection{Complexity of Algorithm \ref{alg:overview}}

We first consider the complexity of Algorithm \ref{alg}, since it is used as sub-routine in Algorithm \ref{alg:overview}. We have efficient algorithms to solve Equation \eqref{eq:predict_enter} and \eqref{eq:predict_leave}, since Equation \eqref{eq:predict_enter} is equivalent to a polynomial (Equation \eqref{eq:enter_poly_form}), and Equation \eqref{eq:predict_leave} is a root finding problem for a monotone functions. For ease of presentation, we regard solving Equation \eqref{eq:enter_poly_form} or Equation \eqref{eq:predict_leave} as one operation.

Let $m=\sum_i d_i$ be the total number of discrete values. Then we have
\begin{lemma}
    \label{lem:alg2_complexity}
    Algorithm \ref{alg} requires at most $2mn$ operations.
\end{lemma}
\begin{proof}
    Each buyer with each possible value at most changes the bidding set twice (the entering point and the leaving point).
    So there are at most $2m$ change points. For each change point,  the algorithm solves at most $n$ equations, either in the form of Equation \eqref{eq:enter_poly_form} or Equation \eqref{eq:predict_leave}.
\end{proof}

\begin{theorem}
    Let $L=\max\{\cup_{i\in N} ~\mathrm{supp}(G_i) \}$. Suppose we use $UB-LB<\epsilon$ as the exit condition. Then Algorithm \ref{alg:overview} requires $O(mn\log(L/\epsilon))$ operations.
\end{theorem}
\begin{proof}
    Algorithm \ref{alg} requires $O(mn)$ operations according to Lemma \ref{lem:alg2_complexity}.
    Algorithm \ref{alg:overview} does a binary search in the interval $(0, L)$, which requires $O(\log(L/\epsilon))$ runs of Algorithm \ref{alg}.
\end{proof}

\subsection{Experiment results}

%
%
%

\subsubsection{Accuracy comparison between continuous algorithms and our algorithm}
When  buyers have discrete value distributions, one natural way of computing the BNE is to approximate the discrete value distribution with a continuous one. Of course, there are infinitely many ways of approximation. Our choice is to replace a discrete value with a ``triangle'' probability density centered at that value, and cover the interval $[v_i^1, v_i^{d_i}]$ with a small uniform distribution. Formally,

\begin{gather*}
	f_i(t)=
	\begin{cases}
		\left[\frac{p_i^d}{w^2}(t-v_i^d+w)\right](1-\epsilon)+\epsilon\cdot u_i & \exists d,t\in[v_i^d-w, v_i^d] \\
		\left[-\frac{p_i^d}{w^2}(t-v_i^d-w)\right](1-\epsilon)+\epsilon\cdot u_i & \exists d,t\in[v_i^d, v_i^d+w] \\
		0 & t<v_i^1 - w \text{~or~} t>v_i^{d_i}+w \\
		\epsilon\cdot u_i & \text{otherwise}
	\end{cases},
\end{gather*}
where $u_i = \frac{1}{v_i^{d_i}-v_i^1+2w}$.

For simplicity, we assume that for all $i$ and $d$, $|v_i^{d+1} - v_i^d| >2w$ and $w<v_i^d<1-w$. Note that although we stick to this approximation throughout this section, our analysis also applies to other possible ways of approximation.

We implemented the continuous backward-shooting algorithm using the characterization by \citet{maskin2000equilibrium}:

\begin{gather}\label{eq:dvdb}
	\frac{\mathrm{d}t_i(b)}{\mathrm{d}b}=
	\frac{F_i(t_i(b))}{f_i(t_i(b)}\left[\left(\frac{1}{n-1}\sum_{j=1}^{n}\frac{1}{t_j(b)-b} \right)-\frac{1}{t_i(b)-b} \right], \forall i\in [n]
\end{gather}

\begin{algorithm}[h]
	\KwIn{step $s$, max winning bid guess $\bar{b}$}
	\KwOut{the min winning bid}
	$b \leftarrow \bar{b}, t_i \leftarrow 1$\;
	\While{$t_i>b, \forall i\in [n]$}{
		\For{$i \in [n]$}{
			compute $t'_i(b)$ according to Equation (\ref{eq:dvdb})\;
			$t_i \leftarrow t_i - t'_i(b)\cdot s$\;
		}
		$b \leftarrow b-s$\;
	}
	\Return $b$\;
	\caption{The continuous backward-shooting algorithm}
	\label{algo:backward}
\end{algorithm}

Using a smaller step size $s$ could significantly increase the number of loops inside Algorithm \ref{algo:backward}. However, using a smaller $s$ could also make the computation result more accurate. Also, when $t_i\not \in \cup_{i,d}[v_i^d-w, v_i^d+w]$, i.e., $f_i(t_i)=\epsilon$, the ratio $\frac{F_i(t_i)}{f_i(t_i)}$ could be very large. Therefore, using a relatively large $s$ could lead to a sudden decrease of $t_i$, causing the program to skip other $v_i^d$'s in between.

The above problem can be avoided by carefully tuning the parameter $s$. However, the possible large values of $\frac{F_i(t_i)}{f_i(t_i)}$ can also cause other problems that may not have easy solutions:
\begin{enumerate}
	\item a large $\frac{F_i(t_i)}{f_i(t_i)}$ leads to a large $t'_i(b)$, meaning that $t(b)$ decreases much faster than $b$. This problem could cause the program to terminate early if $t(b)$ becomes smaller than $b$.
	\item in the next loop, $t'_i(b)$ can be negative and $t_i$ will oscillate as a result (see Figure~\ref{fig:oscillation}).
\end{enumerate}

To understand the first problem, consider the following example:

\begin{example}
Consider the case where there are 6 buyers. Their value distributions are as follows:
\begin{gather*}
(v_1^1, v_1^2, v_1^3) = (0.08, 0.2, 0.8), (p_1^1, p_1^2, p_1^3) = (0.2, 0.76, 0.04),\\
(v_2^1, v_2^2, v_2^3) = (0.09, 0.3, 0.9), (p_2^1, p_2^2, p_2^3) = (0.3, 0.36, 0.34),\\
(v_3^1, v_3^2, v_3^3) = (0.07, 0.12, 0.7), (p_3^1, p_3^2, p_3^3) = (0.3, 0.36, 0.34),\\
(v_4^1, v_4^2, v_4^3) = (0.07, 0.12, 0.7), (p_4^1, p_4^2, p_4^3) = (0.3, 0.36, 0.34),\\
(v_5^1, v_5^2, v_5^3) = (0.07, 0.12, 0.7), (p_5^1, p_5^2, p_5^3) = (0.2, 0.15, 0.65),\\
(v_6^1, v_6^2, v_6^3) = (0.04, 0.12, 0.8), (p_6^1, p_6^2, p_6^3) = (0.2, 0.15, 0.65).
\end{gather*}

\begin{figure}[h!]
	\centering
	\includegraphics[width=0.5\linewidth]{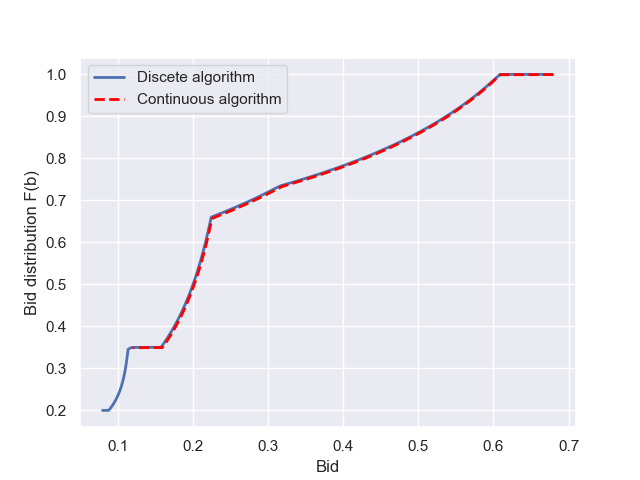}
	\caption[Early termination of continuous algorithm]{}
	\label{fig:earlytermination}
\end{figure}

Figure \ref{fig:earlytermination} shows the bidding strategy of Buyer 5 in the above example. The minimum winning bid computed by the continuous algorithm is about 0.12, while the actual minimum winning bid is 0.08, indicating that Algorithm \ref{algo:backward} terminates early. The reason is that during the execution of Algorithm \ref{algo:backward}, when $b$ is near 0.12, $t_2$ is near 0.25, but $t'_2$ is over 1500. This means that a slight decrease in $b$ could lead to a significant drop in $t_2$, making $t_2<b$ and terminating the algorithm.
\end{example}

To understand the second problem, consider the case where $t_i$ is close to $b$ in the BNE for some $b$. Then $t'_i$ will very likely to be negative according to Equation (\ref{eq:dvdb}), and it is not clear how we could avoid such problems since the computation of $t'_i$ is independent of $s$.

\begin{example}
\label{ep2}
 Consider the case where there are 3 buyers. Their value distributions are as follows:
\begin{gather*}
(v_1^1, v_1^2, v_1^3) = (0.1, 0.2, 0.25), (p_1^1, p_1^2, p_1^3) = (0.25, 0.25, 0.5),\\
(v_2^1, v_2^2, v_2^3) = (0.1, 0.2, 0.25), (p_2^1, p_2^2, p_2^3) = (0.05, 0.45, 0.5),\\
(v_3^1, v_3^2, v_3^3) = (0.1, 0.2, 0.25), (p_3^1, p_3^2, p_3^3) = (0.05, 0.45, 0.5).
\end{gather*}
\end{example}
Figure~\ref{fig:oscillation} shows the bidding strategy of Buyer 1 in BNE. The bid distribution $F(b)$ is computed according to the corresponding $t(b)$. The oscillation in the curve indicates the oscillation in $t(b)$. As shown in the figure, the oscillation occurs when $v(b)$ is close to $b$.
\begin{figure}[h!]
	\centering
	\includegraphics[width=0.5\linewidth]{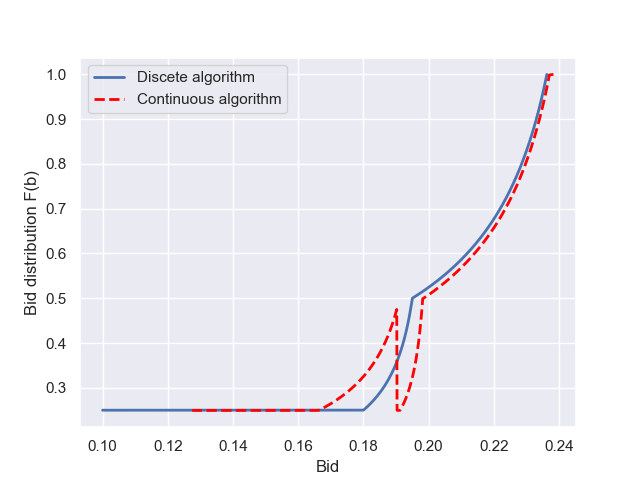}
	\caption{The oscillation of continuous algorithms. The two curves corresponds to buyer 1's equilibrium bidding strategy computed by both the continuous and the discrete algorithms.}
	\label{fig:oscillation}
\end{figure}

Although the continuous backward-shooting algorithm can be problematic, our experiments show that it has good performance if we only need to figure out the maximum winning bid $\bar{b}$. The problems we mentioned mainly affects the computation of the strategies for smaller values. The computation of this part of strategies also suffers from sensitivity issues, as discussed in \cite{fibich2011numerical}. To overcome this difficulty, \citet{fibich2011numerical} proposed another algorithm where guessing the maximum winning bid is no longer needed. We also try their algorithms and conduct experiments. When it comes to approximating discrete value distributions with continuous ones, the computation still suffers from sensitivity issues due to very high condition numbers, implying that the sensitivity might be an intrinsic issue of this problem.

However, our discrete algorithm does not use continuous distributions to approximate discrete ones, thus can avoid all the above problems. The computational complexity only depends on the number of value points. Also, our algorithm does not have the oscillation problem or sensitivity issues, since our theoretical analysis already characterize the structure of the solution, and involves none of the sensitivity computation mentioned above. Therefore, our discrete algorithm can provide a much accurate solution compared to other ones. Such a high accuracy enables us to perform further researches related to first price auctions (see 
Section~\ref{sec:wel_comparison} for example).

\subsubsection{Running time comparison between continuous algorithms and our algorithm}
In this section, we compare the running time of previous continuous algorithms and our discrete algorithm. Since the algorithm provided by \citet{fibich2011numerical} often gives a condition number issue, we only compare our algorithm with Algorithm \ref{algo:backward} in these experiments. We conduct experiments for three different settings. For each setting, the experiment setup is as follows: We generate 1000 first price auction instances, with each containing $n$ buyers. For each buyer, we sample $d$ different values from the interval $[0,1]$, and the corresponding value distribution is also randomly generated for each buyer. Then both our algorithm and Algorithm \ref{algo:backward} are applied to compute the Bayes-Nash equilibrium. Both these two algorithms need to guess the maximum winning bid, so the final computed minimum winning bid would be different from the actual minimum winning bid. Therefore we also set a tolerance parameter $tol$, which serves as a stopping criterion (i.e., the algorithm terminates when the difference between the computed minimum winning bid and the actual minimum winning bid is smaller than $tol$). As these algorithm runs, we record the running time of the algorithms on each instance. Considering that in some cases, the algorithms may take a very long time to terminate, we set another deadline parameter $T$, and kill the process once the running time exceeds $T$. During the experiments, we make sure that no other programs are running and at any time, only one algorithm is running on one instance. Also, we only compare the running time in these experiments, so detailed solution qualities are ignored.

The parameters for the three settings are as follows:
\begin{itemize}
    \item small: $n=5, d=5, T=30\text{~seconds}$;
    \item medium: $n=10, d=10, T=60\text{~seconds}$;
    \item large:$n=100,d=100, T=60\text{~seconds}$;
\end{itemize}
For all the settings, we run our algorithm with $tol=10^{-8}$, and run Algorithm \ref{algo:backward} twice with $tol=0.1$ and $tol=0.01$. The experiment results are shown in Figure

\begin{figure}[ht]
    \centering
    \begin{subfigure}[t]{0.45\linewidth}
        \centering
        \includegraphics[width=\linewidth]{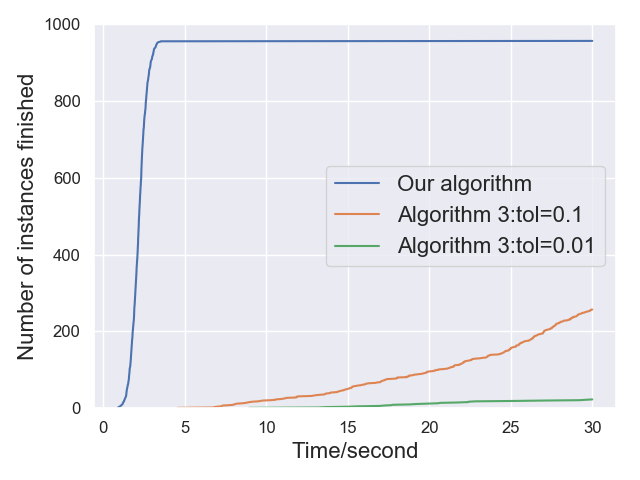}
        \caption{Small instances}
        \label{fig:small_instance}
    \end{subfigure}
    \begin{subfigure}[t]{0.45\linewidth}
        \centering
        \includegraphics[width=\linewidth]{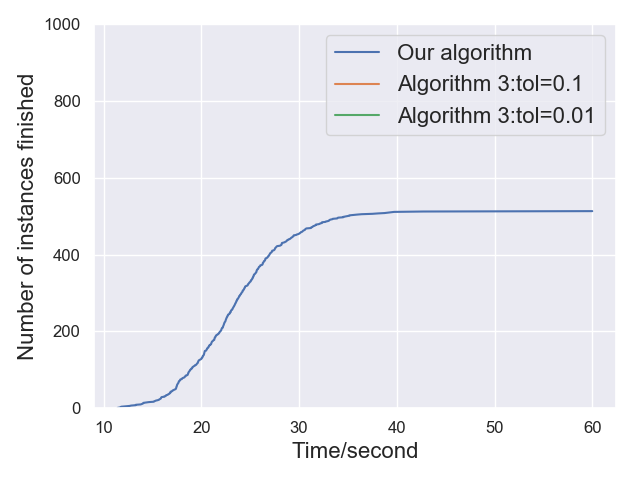}
        \caption{Medium instances}
        \label{fig:medium_instance}
    \end{subfigure}
    \caption{Running time comparison. The y-axis is the cumulative distribution of the running time for the 1000 instances (i.e., the number of instances finished within the corresponding time period). For large instances, no algorithm can finish the computation within the 60 seconds deadline, thus are not shown.}
\end{figure}

For the 1000 small instances, although we set a much smaller tolerance value $tol=10^{-8}$ for our algorithm, our algorithm finishes on almost all instances (955) within the 30 seconds deadline, Algorithm \ref{algo:backward} finishes on only 256 instances when $tol=0.1$ and on only 22 instances when $tol=0.01$. For medium instances, our algorithm finishes on 512 of them, while Algorithm \ref{algo:backward} does not finish on any instance within the deadline. And for large instances, no algorithm ever finishes on any instance within the deadline. It is interesting that our algorithm either finishes very quickly, or does not finish after a relatively long time. For example, among the finished 955 small instances, almost all of them finish within the first 3.5 seconds. This is also true for medium instances. The reason behind this observation is still unknown. We believe this is closely related to specific value distributions in the instances.

\subsubsection{Welfare comparison between first and second price auctions}\label{sec:wel_comparison}

In this section, we compare the welfare of the first price auction and the second price auction, which we denote by $Wel^f$ and $Wel^s$, respectively. The ratio $Wel^f/Wel^s$ is called the price of anarchy:
\begin{definition}[Price of anarchy (PoA)]
	The price of anarchy of the first price auction is defined as the minimum ratio between the welfare of the first price auction and that of the second price auction, i.e.,
	\begin{gather*}
		PoA=\min\frac{Wel^f}{Wel^s},
	\end{gather*}
	where the minimization is over all possible value distributions.
\end{definition}

For general n players, \cite{hartline2014price} shows an example with $PoA=0.869$ and
 \cite{hoy2018tighter} proves that PoA is at least 0.743.
We try to answer this question by running our algorithm on approximate discrete value distributions.
We only consider the problem with two players. Our result provides evidence, though in a limited searching space, that the correct ratio would be around $0.869$.

There are infinitely many value distributions, thus it is impossible to enumerate. However, we could just search for the smallest ratio $Wel^f/Wel^s$ in a much smaller discrete space as follows:
\begin{gather*}
v_i^d = \frac{d}{D}, \forall 0\le d\le D, i=1,2,\\
p_i^d\in\left\{\frac{m}{M}\mid 0\le m\le M, m\in \mathbb{N} \right\}, \forall 0\le d\le D, i=1,2,\\
\sum_{d}v_i^d=1, i=1,2.
\end{gather*}

For simplicity, we normalize the value distributions so that the maximum possible value is 1. It is easy to see that there are only $\frac{1}{2}\binom{M+d}{d}\left[ \binom{M+d}{d}+1\right]$ distinct value distribution pairs.

According to \citet{lebrun2002continuity}, if two value distributions are ``close'', then the corresponding bidding strategies and the resulting welfare values are also ``close''. Therefore, the minimum ratio in the above discrete space should shed us with some insight about what the actual continuous distribution that minimizes the welfare ratio.

In our experiments, we choose $d = 6$ and $M=10$. Thus there are $\frac{1}{2}\binom{16}{6}\left[ \binom{16}{6}+1\right]=32068036$ value distribution pairs to consider. We computed the ratio $Wel^f/Wel^s$ for all of them, and the smallest one ($Wel^f/Wel^s=0.89638$) is given by the following value distributions:
\begin{gather*}
p_1^5=1.0,\\
p_2^0=0.6, p_2^2=0.2, p_2^3=0.2.
\end{gather*}

In this example, Player 1 has a constant large value while Player 2 has a low value distribution.
It coincides with the example introduced by \cite{hartline2014price}, i.e.,  there is a player with a large constant value, while the other players have the identical distribution over a low value range.

Therefore, we make the conjecture that in the extreme case where PoA achieves its smallest value, there exists a player with constant value which is larger than the other players' values.

\section{Conclusion} 
In this paper, we consider the problem of computing the BNE of the first price auction for discrete value distributions. From the theoretical point of view, we first propose Algorithm 2 to compute it. Contrary to previous algorithms for the continuous value case, our algorithm does not need to solve the ordinary differential equations. The number of computation operations is at most $2nm$ in Algorithm \ref{alg}. Our algorithm is enabled by the introduction of a useful tool called the bidding set and careful analysis of the detailed structure of the BNE of the first price auction. 

Then, we show the BNE is unique without any assumptions on the value distributions.
From the practical point of view, Algorithm 2 is much faster than the traditional algorithms when computing BNE under the same discrete distribution. Algorithm 2 is also more robust on providing the strategies on small values shown in Example \ref{fig:earlytermination}.  
Thus our algorithm has a great advantage over the traditional algorithms for practical applications.

With such an algorithm to compute the BNE, we can also explore other problems related to first price auctions. For example, we compare the revenue generated in the first and the second price auctions. It shows that in extreme cases, each auction can perform much better than the other. The price of anarchy of the first price auction is also measured in a set of approximate distributions, we give a conjecture about the extreme case based on our experiments and observations.


\bibliographystyle{plainnat}
\bibliography{bne}

\newpage
\appendix
\section*{APPENDIX}
\setcounter{section}{0}
\section{Omitted proof in Section \ref{sec:bidding_set}}
\label{appendix:interval_proof}
\subsection{Proof of Theorem~\ref{thmbidcont}}
To prove Theorem~\ref{thmbidcont}, We first consider a lemma which compare buyers' value from the bid distribution.
\begin{lemma}
	Assume bid $b_1$ is an element in $S_i(v_i)$, and bid $b_2$ is an element in $S_j(v_j)$. Suppose $b_1<b_2$.
	If there is zero probability that buyer $i$ bids in $(b_1,b_2)$, then we have $v_j\geq v_i$.
	Furthermore, if $v_j= v_i$, then there is zero probability that buyer $j$ bids in $(b_1,b_2)$ too.
	\label{2pmonotone}
\end{lemma}
\begin{proof}
	Since $b_1$ is in $S_i(v_i)$, $b_1$ is a best response for buyer $i$ with value $v_i$.
	\begin{eqnarray*}
		(v_i-b_1)\cdot F_{\Omega\setminus i}(b_1)&\geq & (v_i-b_2)\cdot F_{\Omega\setminus i}(b_2)\\
		(v_j-b_2)\cdot F_{\Omega\setminus j}(b_2)&\geq&(v_j-b_1)\cdot F_{\Omega\setminus j}(b_1)
	\end{eqnarray*}
	Multiply two equations and get
	$$	\frac{v_i-b_1}{v_i-b_2}\cdot \frac{F_i(b_2)}{F_i(b_1)}\geq\frac{v_j-b_1}{v_j-b_2}\cdot \frac{F_j(b_2)}{F_j(b_1)}$$
	
	By Lemma \ref{nomass}, there is no mass bid in buyer $i$'s bid distribution between $(b_1, b_2)$, i.e., $F_i(b_2)=F_i(b_1)$. 
	Since $F_j(b_2)\geq F_j(b_1)$, we have
	\begin{eqnarray*}
		\frac{v_i-b_1}{v_i-b_2}&\geq&\frac{v_j-b_1}{v_j-b_2}\\
		v_j&\geq&v_i
	\end{eqnarray*}
	If $v_j=v_i$, all the inequalities should be equality. Then we have $F_j(b_2)=F_j(b_1)$.
\end{proof}

	\begin{proof}
		We prove the theorem by contradiction.
		Assume there is a value with a jump in its bid set. By jump, we mean the closure of the support of $S_i$ is not connected.
		Let $v_i^k = \max_{j,l} \{v^l_j \mid S_j(v^l_j) \textrm{ has a jump.}\}$.
		Since there is a jump, we assume $b_1,b_2\in S_i(v^k_i)$, and $(b_1,b_2)\cap S_i(v^k_i)=\emptyset$.
		Assume buyer $j$ with value $v^l_j$ has bid in this interval $(b_1,b_2)$.
		By Lemma \ref{2pmonotone}, we have $v_j^l\geq v^k_i$.
		If $v^l_j=v^k_i$, then buyer $j$ with value $v^l_j$ has zero probability bidding in the interval $(b_1,b_2)$.
		There is no impact when the probability is zero. So we can regard this case as buyer $j$ does not bid in $(b_1,b_2)$.

		If $v^k_j>v_i$, by the assumption that $v^k_i$ is the largest value that has a gap in the support, we know $S_j(v_j^l)$ has no jump.
		$S_j(v_j^l)$ is a connected interval. So we can get $S_j(v_j^l)$ from $S_j(v_j^l)$ after removing possibly infinite points.
		
		
		Let set $\Lambda_1$ denote the buyers who bid in the upper neighborhood of $b_1$ and $\Lambda_2$ denote the buyers who bid in the lower neighborhood of $b_2$. Formally, we have 
		$$\Lambda_1=\left\{j \mid \exists v^{l}_j, \epsilon>0 \textrm{ s.t. } S_j(v_j^l)\cap (b_1,b_2) \neq \emptyset \textrm{ and } (b_1,b_1+\epsilon)\subset S_j(v_j^l)\right\},$$
		$$\Lambda_2=\left\{j \mid \exists v^{l}_j, \epsilon>0 \textrm{ s.t. } S_j(v_j^l)\cap (b_1,b_2) \neq \emptyset \textrm{ and } (b_2-\epsilon,b_2)\subset S_j(v_j^l)\right\}.$$		
		Buyers who bid in $(b_1,b_2)$ have connected bidding intervals.
		By Lemma~\ref{lemma4}, the union of these connected bidding intervals cover $(b_1,b_2)$.
		So set $\Lambda_1,\Lambda_2$ are well defined and not empty.

		For any $j\in \Lambda_1$, we claim $b_2\in S_j=\bigcup_h S_j(v_j^h)$. Otherwise, $\overline{S_j}$ has jump around $b_2$. Then we define
		$$v_j^*=\max \left\{v \mid S_j(v)\subseteq [0,b_2)\right\}.$$
		By Lemma~\ref{2pmonotone}, we have $v_i^k\geq v_j^*$.
		By the monotonicity property (Lemma~\ref{1pmonotone}), we have $v^*_j\geq v_j$.
		Thus we have $v_i^k\geq v_j^*\geq v_j>v_i^k$, a contradiction.
		
		So for any $j\in \Lambda_1$, we have $b_2\in S_j$ which implies $j\in \Lambda_2$. Thus $\Lambda_1$ is a subset of $ \Lambda_2$. 
		Since bid $b_1+\epsilon$ does not generate more utility than bid $b_1$, for buyer $i$ with value $v_i^k$, the derivative of $(v_i^k-x)\prod_{j\in \Lambda_1} F_j(x)$ at $x=b_1$ is non-positive. Taking derivatives, we get
		$\sum_{j\in \Lambda_1}\frac{f_j(b_1)}{F_j(b_1)}\leq \frac{1}{v_i-b_1}$. By Theorem~\ref{thmeqa}, we should have
		$$\frac{1}{|\Lambda_1|-1}\sum_{j\in \Lambda_1}\frac{1}{v_j(b_1)-b_1}\leq \frac{1}{v_i^k-b_1}.$$
		Since buyer $i$ weakly prefer bid $b$ rather than $b_2-\epsilon$, similarly we should have
		$$\frac{1}{|\Lambda_2|-1}\sum_{j\in \Lambda_2}\frac{1}{v_j(b_2)-b_2}\geq \frac{1}{v_i-b_2}.$$
		Based on these two equations, we have,
		\begin{align*}
		|\Lambda_2|-1&\leq\sum_{j\in \Lambda_2}\frac{v_i-b_2}{v_j(b_2)-b_2}\\
		&\leq|\Lambda_2|-|\Lambda_1|+\sum_{j\in \Lambda_1}\frac{v_i-b_2}{v_j(b_2)-b_2}\\
		&<|\Lambda_2|-|\Lambda_1|+\sum_{j\in \Lambda_1}\frac{v_i-b_1}{v_j(b_2)-b_1}\\
		&\leq|\Lambda_2|-|\Lambda_1|+\sum_{j\in A}\frac{v_i-b_1}{v_j(b_1)-b_1}\\		
		&\leq|\Lambda_2|-|\Lambda_1|+|\Lambda_1|-1.
		\end{align*}
		A contradiction.
	\end{proof}
	


%
%

\subsection{Proof of Theorem \ref{strict_decrease}}
\begin{proof}
The proof consists of two parts. Lemma \ref{virtualvalue} solves the case when $\Lambda(x)$ is fixed.  
Lemma \ref{lemmavstarchange} solves the case when $\Lambda(x)$ changes.
\end{proof}

\begin{lemma}
	\label{virtualvalue}
	When $\Lambda(x)$ does not change and the buyer's value in bidding set keep the same, $\Lambda^*(x)$ strictly decreases as x decreases.
\end{lemma}

		\begin{proof}
			For any $j\in \Lambda$, we have
			\begin{eqnarray*}
				\frac{1}{|\Lambda|-1}\sum_{i\in \Lambda}\frac{1}{v_i-x}-\frac{1}{v_j-x}&\geq& 0,\\
				\frac{1}{\phi^*(x)-x}&\geq& \frac{1}{v_j-x}.
			\end{eqnarray*}
			$$\frac{1}{|\Lambda|-1}\sum_{i\in \Lambda}\frac{1}{v_i-x}=\frac{1}{\phi^*(x)-x}.$$
			Recall that $\Lambda(x)$ is fixed.
			By definition of $\phi^*$, it is differentiable. We take derivatives on the both sides,
            \begin{gather*}
            \frac{1}{|\Lambda|-1}\sum_{i\in \Lambda}\frac{1}{(v_i-x)^2}=\frac{1-(\phi^*)'(x)}{(\phi^*(x)-x)^2},\\
            [\phi^*(x)-x]^2\cdot \frac{1}{|\Lambda|-1}\sum_{i\in \Lambda}\frac{1}{(v_i-x)^2}=1-(\phi^*)'(x).
            \end{gather*}
			We want to prove $(\phi^*)'(x)>0$, it is equivalent to prove
			\begin{align}
			&[\phi^*(x)-x]^2\cdot \frac{1}{|\Lambda|-1}\sum_{i\in \Lambda}\frac{1}{(v_i-x)^2}<1,\nonumber\\
			\Leftrightarrow&\frac{1}{|\Lambda|-1}\sum_{i\in \Lambda}\frac{1}{(v_i-x)^2}<\frac{1}{(|\Lambda|-1)^2}\left(\sum_{i\in \Lambda}\frac{1}{v_i-x}\right)^2,\nonumber\\
			\Leftrightarrow&(|\Lambda|-1)\sum_{i\in \Lambda}\frac{1}{(v_i-x)^2}<\left(\sum_{i\in \Lambda}\frac{1}{v_i-x}\right)^2.\label{eq2}
			\end{align}
			By Theorem~\ref{thmeqa}, for any $j\in \Lambda$,  we have
			$$(|\Lambda|-1)\frac{1}{(v_j-x)^2}\leq \left(\sum_{i\in \Lambda} \frac{1}{v_i-x}\right)\cdot \frac{1}{v_j-x}.$$
			Since it is impossible that for all $j\in \Lambda$ it is equality, so the sum of inequalities is a strict inequality. Thus we prove Eq. (\ref{eq2}).

			%
		\end{proof}
\begin{lemma}
	\label{lemmavstarchange}
	In all possible changes, $\phi^*(x)$ weakly decreases.
\end{lemma}

	\begin{proof}
	We consider the general case at bid $b$.
	Let $\Lambda_1=\Lambda^+(b)\setminus \Lambda^-(b)$,
	$\Lambda_2=\Lambda^+(b)\cup \Lambda^-(b)$, and
	$\Lambda_3=\Lambda^-(b)\setminus \Lambda^+(b)$.
	
	We want to prove 
	\begin{gather*}
		\frac{1}{|\Lambda_1 \cup \Lambda_2|-1}\sum_{i\in \Lambda_1 \cup \Lambda_2}\frac1{\phi^*_i(b+\epsilon)-b}\leq
		\frac{1}{|\Lambda_2 \cup \Lambda_3|-1}\sum_{i\in \Lambda_2 \cup \Lambda_3}\frac1{\phi^*_i(b-\epsilon)-b}.
	\end{gather*}

	By Lemma \ref{lemmachange}, we have 
	\begin{align}
	\frac1{\phi^*_i(b-\epsilon)-b}&\geq
	\frac{1}{|\Lambda_1 \cup \Lambda_2|-1}\sum_{j\in \Lambda_1 \cup \Lambda_2}\frac1{\phi^*_j(b-\epsilon)-b}\label{eq5}, \forall i\in \Lambda_3\\
	\frac1{\phi^*_i(b-\epsilon)-b}&\leq
	\frac{1}{|\Lambda_1 \cup \Lambda_2|-1}\sum_{j\in \Lambda_1 \cup \Lambda_2}\frac1{\phi^*_j(b-\epsilon)-b}, \forall i\in \Lambda_1\label{eq6}
	\end{align}
	Sum Equation \eqref{eq5} over $i\in \Lambda_3$, and we have
	$$\sum_{i\in \Lambda_3}\frac1{\phi^*_i(b-\epsilon)-b}\geq 
	\frac{|\Lambda_3|}{|\Lambda_1 \cup \Lambda_2|-1}\sum_{j\in \Lambda_1 \cup \Lambda_2}\frac1{\phi^*_j(b-\epsilon)-b}$$
	
	Sum Equation \eqref{eq6} over $i\in \Lambda_1$, and we have
	$$\sum_{i\in \Lambda_2}\frac1{\phi^*_i(b-\epsilon)-b}\geq 
	\frac{|\Lambda_2|-1}{|\Lambda_1 \cup \Lambda_2|-1}\sum_{j\in \Lambda_1 \cup \Lambda_2}\frac1{\phi^*_j(b-\epsilon)-b}.$$
	Combining the two inequalities above completes the proof.

\end{proof}
\subsection{Proof of Theorem \ref{thmenter}}
\begin{proof}
	
	We prove the theorem by contradiction. Suppose $b$ is the largest point where there is a buyer $i$, who satisfies the condition of entering the bidding set, but in BNE he enters the bidding set at some later point, say $y$.
	By the Lemma \ref{lemmachange}, we have $\frac{1}{\phi^{*}(x^-)-x^-}\geq \frac{1}{v_i-x^-}$ in BNE.
	Then $\phi^*(x^{-})\leq v_i$.
	By Theorem \ref{strict_decrease}, virtual value strictly increases, i.e., $\phi^*(y)<\phi^*(x^-), y< x^-$.
	We have $\frac{1}{\phi^*(y)-y}>\frac{1}{v_i-y}$. A contradiction.
	\end{proof}	

\subsection{Proof for Theorem~\ref{thm:mono}}
	\begin{proof}
		We guess two largest winning bids $\overline{b}^1 < \overline{b}^2$.
	Next we transform $\mathcal{E}(\overline{b}^1, \{p_i^j\})$ into $\mathcal{E}(\overline{b}^2, \{p_i^j\})$ through intermediate steps with changing the probability parameters.

	In the structure $\mathcal{E}(\overline{b}^1, \{p_i^j\})$, we sort the bid intervals by their lower extreme points denoted by $l_k=\left[\underline{l}_k,\overline{l}_k\right], k\in[\sum_i d_i]$ where $\underline{l}_1\geq \underline{l}_2\geq ...\geq \underline{l}_{\sum_i d_i}$.
	We use $buyer(l_k)$ to indicate the buyer who gives this bid interval and use $index(l_k)$ to indicate the value index of $buyer(l_k)$. Formally, we have $S_{buyer(l_k)}\left(v_{buyer(l_k)}^{index(l_k)}\right)= \left[\underline{l}_k,\overline{l}_k\right]$.
	
	By changing the probability parameter $\{p_i^j\}$, 
	we can control where a buyer leaves the bidding set according to Theorem \ref{thmrule1leave}.
	We create a parameter profile $\{p^j_i\}^{(1)}$ such that 
	$\mathcal{E}(\overline{b}^2, \{p^j_i\}^{(1)})$ and $\mathcal{E}(\overline{b}^1, \{p^j_i\}^{(1)})$ are same when bid intervals are restricted to the $\left[-\infty,\overline{b}^1\right]$ area.
	We can do this by increasing $p^{d_i}_i$ for $i\in \Lambda^1(\bar{b}^1)$ such that 
	every buyer in $\Lambda^1(\bar{b}^1)$ are still bidding at $\bar{b}^1$ under his largest value.
	Here we use $\Lambda^1$ to denote the bidding set in the structure $\mathcal{E}(\overline{b}^1, \{p^j_i\}^{(1)})$.
	
	
	Set $k=1$ at first. We choose a parameter profile $\{p_i^j\}^{(k+1)}$ satisfying two conditions.
	Condition 1 is that the bid intervals in $\mathcal{E}(\overline{b}^2, \{p_i^j\}^{(k+1)})$ and $\mathcal{E}(\overline{b}^2, \{p_i^j\})$ are same when they are restricted to the $\left[\underline{l}_k,\overline{b}^2\right]$ area.
	Condition 2 is that every corresponding bid intervals in $\mathcal{E}(\overline{b}^2, \{p_i^j\}^{(k)})$ and $\mathcal{E}(\overline{b}^1, \{p^j_i\})$
	have same lower extreme points when they are restricted to the $\left[-\infty,\underline{l}_k\right]$ area.
	By the second condition, $\left\{p^{value(l_k)}_{buyer(l_k)}\right\}^{(k)}$ returns to $\{p_i^j\}$.
	
	By changing the parameter, the bid interval of $buyer(l_k)$ with $index(l_k)$-th value ends earlier. 
	For any bid $x$, $\frac{1}{|\Lambda(x)|-1}\sum_{i\in \Lambda(x)}\frac{1}{v_i(x)-x}$ increases when any buyer in the bidding set $\Lambda(x)$ leaves.
	This change has two impacts. 
	First, bid density increases in the range where $buyer(l_k)$ leaves. 
	The position that bid probability achieves the value probability becomes higher.
	Second, because the bidding set changes, the location where a buyer enters the bidding set may change too.
	Remember we keep the lower extreme point of every bid intervals unchanged.
	Thus there is only one possible change in the bidding set where a buyer enters, i.e.,  $buyer(l_k)$ leaves.
	The value of $\frac{1}{|\Lambda(x)|-1}\sum_{i\in \Lambda(x)}\frac{1}{v_i(x)-x}$ increases. 
	According to Theorem \ref{thmenter}, buyer enters bidding set earlier.
 	Only bid intervals  $l_m$ for $m\geq k+1$ that have smaller lower extreme points have been affected. 
	As a result, $\{p^j_i\}^{(k+1)}\geq \{p^j_i\}^{(k)}$ for $(i,j)$ not in $\{(buyer(l_m),value(l_m)), m=1,..k\}$. 
	
	We repeat the above procedure and increase $k$ by 1 each time. Bid intervals only weakly increase each time. Finally, we would have $\{p^j_i\}^{(\sum_i d_i)}= \{p^j_i\}$ and the proof completes.
	\end{proof}
\subsection{Proof for Theorem \ref{continuity}}
	\begin{proof}

	Consider the case that $\overline{b}$ gets close to $\overline{b}^1$ from the smaller side, we show the limit of structure equals to the structure of the limit, i.e., 
	$$\lim_{\overline{b}\rightarrow \overline{b}^1} \mathcal{E}(\overline{b}, \{p_i^j\})=\mathcal{E}(\overline{b}^1, \{p_i^j\}).$$
	We prove by contradiction. The proof for the other case is similar.
	
	Pick $\overline{b}^2$ close enough to $\overline{b}^1$ such that all the boundaries of bid intervals in 
	$\mathcal{E}(\overline{b}^2,\{p_i^j\})$ 	
	are $\epsilon$ close the $\lim_{\overline{b}\rightarrow \overline{b}^1} \mathcal{E}(\overline{b}, \{p_i^j\})$.
	In the limit of structure, some intervals starts later of ends later compared to structure of limit.
	We choose the largest bid position $x$ where the limit of structure is not equal to the structure of $\overline{b}^1$. 

	Since the upper part $\left[x,\overline{b}^2\right]$ is exactly same, it is impossible that the difference between two structures is some bid interval in the limit ends later.
	It must be the case that some bid interval starts later. 
	Let buyer i has the largest corresponding value $v_i$ among bid intervals that start later.
	So at the moment when buyer i joins the bidding set, the bidding sets are same in $\mathcal{E}(\overline{b}^2, p_i^j)$ and $\mathcal{E}(\overline{b}^1, \{p_i^j\})$.
	
	But buyer i with $v_i$ does not join the bidding set at $x$ in $\mathcal{E}(\overline{b}^2, \{p_i^j\})$, 
	while buyer i joins the bidding set at $x$ in $\mathcal{E}(\overline{b}^1, \{p_i^j\})$. 
	It contradicts with Theorem~\ref{thmenter}.
	
%
%
%
%

	\end{proof}
	

\section{Omitted proofs in Section \ref{sec:exist}}
\subsection{Proof of Lemma \ref{lem:match}}
\begin{proof}
	The algorithm stops only when $|\Lambda|\leq 1$ and no other buyers can join the bidding set.
	By definition of the smallest winning bid, we have $v^1_1\geq \underline{b}$.
	Since $v^1_1$ is buyer 1's smallest value, so buyer 1 never leaves bidding set when he has value $v^1_1$.
	Hence, the bidding set only contains buyer 1 at the end point. Other buyers' next values are  smaller than or equal to the value at end point. Then we put the remaining bid probability of buyer 1 on bidding $\underline{b}$. 
	We are also able to create a bidding strategy 
	For other buyers when their value smaller than or equal to $\underline{b}$, we create a bidding strategy such that they bid $\underline{b}$ deterministically. In this case buyer 1 does not have incentive to deviate from bidding $\underline{b}$ and other buyers never win in the following second price auction because their values are smaller.
\end{proof}

\end{document}